\documentclass[finalsubmission]{IEEEtran}
\usepackage{booktabs,threeparttable}
\usepackage{amsmath}
\usepackage{amssymb}
\usepackage{amsthm}
\usepackage{graphicx}
\usepackage{subcaption}
\usepackage{thmtools}

\newtheoremstyle{plain}
  {0pt} 
  {0pt} 
  {\itshape} 
  {} 
  {\bfseries} 
  {.} 
  {.5em} 
  {} 
\newtheorem{thm}{Theorem}[section]
\newtheorem{defn}{Definition}[section]
\newtheorem{prob}{Problem}[section]

\usepackage{graphics} 
\usepackage{epsfig} 
\usepackage{times} 
\usepackage{algorithm}
\usepackage[noend]{algorithmic}
\usepackage{xcolor}

\usepackage[doi=false,isbn=false,url=false,eprint=false, sorting=none]{biblatex}
\usepackage{tikz}
\usepackage{lipsum}

\newcommand\copyrighttext{
    \footnotesize \copyright 2023 IEEE. Personal use of this material is permitted. Permission from IEEE must be obtained for all other uses, in any current or future media, including reprinting/republishing this material for advertising or promotional purposes, creating new collective works, for resale or redistribution to servers or lists, or reuse of any copyrighted component of this work in other works.}
\newcommand\copyrightnotice{%
\begin{tikzpicture}[remember picture,overlay]
\node[anchor=south,yshift=5pt] at (current page.south) {\fbox{\parbox{\dimexpr\textwidth-\fboxsep-\fboxrule\relax}{\copyrighttext}}};
\end{tikzpicture}}

\title{\LARGE \bf
Safety verification of Neural-Network-based controllers:\\ a set invariance approach
}

\author{Louis Jouret, Adnane Saoud and Sorin Olaru
\thanks{This work is supported by the ANR PIA funding: ANR-20-IDEES-0002.}
\thanks{Louis Jouret is with Swiss Federal School of Technology in Lausanne-EPFL, Switzerland, {\tt\small louis.jouret@epfl.ch}}
\thanks{Adnane Saoud is with the College of Computing, University Mohammed VI Polytechnic, Benguerir, Morocco {\tt\small adnane.saoud@um6p.ma}}%
\thanks{Sorin Olaru is with CentraleSupelec, University Paris-Saclay, Gif-sur-Yvette,  France, {\tt\small sorin.olaru@centralesupelec.fr}%
}
\thanks{Digital Object Identifier (DOI): 10.1109/LCSYS.2023.3342088}}
\makeatletter

\begin{document}
\title{Safety verification of Neural-Network-based controllers: a set invariance approach}

\maketitle

\begin{abstract}
This paper presents a novel approach to ensure the safety of continuous-time linear dynamical systems controlled by a neural network (NN) based state-feedback. Our method capitalizes on the use of continuous piece-wise affine (PWA) activation functions (e.g. ReLU) which render the NN a PWA continuous function. By computing the affine regions of the latter and applying Nagumo's theorem, a subset of boundary points can effectively verify the invariance of a potentially non-convex set. Consequently, an algorithm that partitions the state space in affine regions is proposed. The scalability of our approach is thoroughly analyzed, and extensive tests are conducted to validate its effectiveness.
\end{abstract}
\copyrightnotice

\section{Introduction}
Machine learning, particularly NNs, has had a transformative impact on various scientific fields, including control systems. Two main approaches have emerged: the first involves approximating a complex control law such as Model Predictive Control (MPC) using a NN \cite{MADDALENA202011362}, making the controller more memory efficient and enabling faster computations. The second approach entails synthesizing a NN controller through reinforcement learning, which has gained popularity due to its capability to learn intricate control strategies from data generated by the system \cite{lillicrap2015continuous}. However, when applying a NN to safety-critical systems \cite{HUNT19921083}, there is often skepticism and valid concerns regarding their black box nature and the inherent difficulty in interpreting their behavior.
In control systems, safety refers to the property of a system remaining in a set of safe states for all future time instances. NNs, characterized by their numerous neurons and nonlinear activation functions, present computational challenges in explicitly representing the input-output relationship. This complexity hampers the ability to interpret the actions taken by a NN-controller and verify its safe operation. Extensive research has been conducted to address this issue using different approaches.
One approach involves estimating the reachable set of the NN-controlled system \cite{huang2019reachnn, julian2019reachability, 9561956}. However, existing reachability-based approaches are limited to discrete-time systems and finite time safety, while in the present work we are dealing with continuous-time systems and infinite-time safety properties. Another approach focuses on finding a barrier certificate \cite{zhao2020synthesizing, zhao2022verifying} for the closed-loop system. The main challenge lies in computing the barrier certificate. Recent advancements involve training a separate NN that acts as barrier certificates to the closed-loop system \cite{zhao2020synthesizing, zhao2022verifying}. This approach leverages the universal approximation capabilities of the NN, enabling them to estimate a barrier certificate if one exists. Nevertheless, a notable limitation of this method is its lack of completeness. In other words, if the process fails to identify a barrier certificate, it remains uncertain whether the failure stems from an inability to discover the correct certificate or the absence of a valid certificate altogether.
This paper presents a novel approach to ensure safety of linear dynamical systems controlled by a NN. The method revolves around Nagumo's condition, which states that a set is invariant if the vector field at every point on the boundary points back inside the set \cite{inbookcite}. By utilizing continuous PWA activation functions within the NN, like $ReLU$, the output of the latter can be explicitly expressed as a continuous composite PWA function. Nagumo's condition is then used for the invariance of linear systems controlled by PWA controllers within a polytopic set. Consequently, the verification of a small subset of boundary points is sufficient to prove the set's invariance. Although the number of regions within which the PWA controller is affine increases non-polynomially with the number of neurons \cite{montufar2014number}, the advantage of the present approach relies on the fact that the calculations are limited to the regions connected to the set's boundaries. To this end, an algorithm that leverages the \textit{automatic differentiation} \cite{baydin2018automatic} offered by modern deep-learning libraries like \textit{PyTorch} \cite{paszke2019pytorch} is described and analysed. 

Section \ref{sec:2} introduces the class of systems and NN considered in the paper. In Section \ref{sec:3}, the condition to guarantee the safety of the considered system is presented. Section \ref{sec:4} proposes an algorithmic procedure to apply the presented method. Finally, Section \ref{sec:5} focuses on the scalability of the proposed approach.

\textbf{Notation:} $\mathbb{R}$ and $\mathbb{R}^+_{0}$ are the set of reals and of non-negative reals, respectively. Given a set $\mathcal{O} \subset \mathbb{R}^n$, its border set is denoted $\partial\mathcal{O}$. For a matrix $M$, we denote $M_{i,j}$ the element in the $i^{th}$ row and $j^{th}$ column and $M_{i,*}$ the $i^{th}$ row vector. The matrix $diag(v)$ designates a diagonal matrix having the scalars $v_i \in \mathbb{R}$ on its diagonal. For $v, w \in \mathbb{R}^n$, we write $v \preccurlyeq w$ the element-wise inequality of two vectors such that $v_i \leq w_i, \forall i$. To represent the multiplication of multiple matrices, denoted as $M_1, M_2, \ldots, M_n$, we introduce the notation ${\overset{\curvearrowleft}\prod^{n}_{k = 1}} M_k$. In this notation, the left arrow signifies that the matrices are multiplied from right to left, starting with $M_n$ and ending with $M_1$, resulting in the final matrix product. Considering a NN with $L$ hidden layers and $N$ neurons per hidden layer, we denote $n_{in}$, $n_{out}$ and $n_{in}\times N^{(L)}\times n_{out}$ the network's input dimension, output dimension and architecture respectively. The number of neurons on layer $l$ will be written as $n^{(l)}$. We write $\binom{n}{k}$ the binomial coefficient $\frac{n!}{k!(n-k)!}$.
\section{Preliminaries and Problem Formulation}
\label{sec:2}
\subsection{Linear Dynamical Systems}
In this paper, we consider a linear system $\Sigma$  described by
\begin{align}
\dot{x} = Ax + Bu
\label{eq: sys}
\end{align}
where $x  \in \mathbb{R}^n$ is the state and $u \in \mathbb{R}^m$ is the control input.
\subsection{Feed-forward Neural Networks}
A feed-forward neural network (NN) consists of interconnected layers, with each layer containing individual neurons \cite{HORNIK1989359}. The connection weights between layer $l-1$ and layer $l$ are denoted as $W^{(l)}$, and layer $l$'s biases are represented as $b^{(l)}$. The output $z^{(l)}$ of layer $l$, of a NN with an input $x$, can be expressed as follows:
\begin{equation*}
    z^{(l)} = 
        \sigma^{(l)}(W^{(l)}z^{(l-1)} + b^{(l)}), \quad l \in \{1,2,..., L\}, \quad z^{(0)}=x,
\end{equation*}
with $z^{(l-1)}\in \mathbb{R}^{n^{(l-1)}}$, $W^{(l)} \in \mathbb{R}^{n^{(l)} \times n^{(l-1)}}$ and $b^{(l)} \in \mathbb{R}^{n^{(l)}}$ and the componentwise \textit{activation function} $\sigma^{(l)} : \mathbb{R}^{n^{(l)}} \rightarrow \mathbb{R}^{n^{(l)}}$ introduces a nonlinear behavior to the neuron's output.
The overall output of the NN is $\mathcal{N}(x; \theta) = z^{(L)}$, where $\theta$ are the parameters of the NN, i.e. the weights and biases. 
\subsection{Convex Polytopes}
This section recalls key concepts related to polytopes \cite{brondsted2012introduction}, highlighting their main geometric properties for this study.
\begin{defn}
    Let $C \in \mathbb{R}^{m\times n}$ and $d \in \mathbb{R}^m$. A closed convex polyhedral set is defined in its $\mathcal{H}$-representation as follows: $\mathcal{R} = \{x\in \mathbb{R}^n \mid Cx \leq d\}$.
\end{defn}
When the inequality is strict (i.e., $Cx < d$), the polyhedron is referred to as open. A polytope is a bounded polyhedron.
\begin{defn}
    Let $\mathcal{V} = \{v_1,\hdots,v_n\}$ be a finite set of points in $\mathbb{R}^d$. The convex hull of $\mathcal{V}$ defines a polytope as follows:
    \begin{align*}
        \mathcal{R}  &= \left\{\sum_{k = 1}^{n} \lambda_i v_i \mid \lambda_i \geq 0, \sum_{k = 1}^{n} \lambda_i = 1, v_i \in \mathcal{V}\right\}
    \end{align*}
    This is also known as the $\mathcal{V}$-representation of a polytope and $\mathcal{V}(\mathcal{R})$ is the set of vertices of the polytope $\mathcal{R}$.
\end{defn}
\begin{defn}
    Let $\mathcal{R} \subset \mathbb{R}^n$ be a convex polytope defined by $\mathcal{R} = \{x \in \mathbb{R}^n\mid Cx\leq d\}$ where $C\in \mathbb{R}^{m\times n}$ and $d\in \mathbb{R}^m$. Consider the hyperplane $\mathcal{H}_i = \{x\in \mathbb{R}^n\mid C_{i,*}x = d_i\}, i \in \{1, \hdots, m\}$.
    The face $\mathcal{F}_i$ of $\mathcal{R}$ is defined as $\mathcal{F}_i(\mathcal{R}) = \mathcal{H}_i \cap \mathcal{R}$ and $\mathcal{F}(\mathcal{R})$ will denote the set of all the faces of $\mathcal{R}$.
\end{defn}
\subsection{Neural Network controlled systems}
The linear system $\Sigma$ in (\ref{eq: sys}) is considered to be controlled by a NN, represented by the state-feedback function $u(x) = \mathcal{N}(x; \theta)$ \cite{HUNT19921083}. The NN receives the system's states and produces the corresponding control input for the system. This class of systems is referred to as NN-controlled systems and has been extensively explored recently \cite{huang2019reachnn,julian2019reachability,zhao2022verifying}. In this case, the closed-loop system is given by
\begin{align}
    \dot{x}(t) = f(x(t)) = Ax(t) + B\mathcal{N}(x(t);\theta)
    \label{eq: sys_NN}
\end{align}
\subsection{Problem formulation}
In this paper, we consider the following problem:
\begin{prob}
\label{prob}
    Consider the closed-loop NN-controlled system in (\ref{eq: sys_NN}). Let \(\mathcal{S}\) be a polytopic set of admissible states in \(\mathbb{R}^n\), and let $\mathcal{O}_1, \mathcal{O}_2, \ldots, \mathcal{O}_m \subset \mathcal{S}$ be $m$ open polytopic unsafe sets. Consider the safe set $\mathcal{X} := \mathcal{S}\setminus \bigcup_{i=1}^{m} \mathcal{O}_{i}$. The objective is to verify that for any trajectory $ x : \mathbb{R}_0^{+} \rightarrow \mathbb{R}^n$ satisfying $x(0) \in \mathcal{X}$, we have that $x(t) \in \mathcal{X}$, for all $t \in \mathbb{R}_0^+$.
\end{prob}
Intuitively, the objective is to prove the positive invariance of the potentially non-convex set $\mathcal{X}$ for the NN-controlled system in (\ref{eq: sys_NN}).

\section{Main Results}
\label{sec:3}
\subsection{Piece-wise affine Neural Networks}
Consider the \textit{activation function} $\sigma: \mathbb{R}\rightarrow\mathbb{R}$ to be PWA:
\begin{align}
    \sigma(x) &=
    \left\{\begin{matrix}
           c_1 \cdot x + d_1 & \text{if } x \leq m_1\\
            c_2 \cdot x + d_2 & \text{if } m_1 < x \leq m_2\\
            \vdots\\
            c_n \cdot x + d_n & \text{if } m_{n-1} < x\\
        \end{matrix}\right.
    \label{activation}
\end{align}
where $c_i, d_i, m_i \in \mathbb{R}$, $i \in \{1,\hdots,n\}$.
Commonly used functions that satisfy this condition are \textit{ReLU}, \textit{leaky-ReLU}, \textit{PReLU} or the \textit{Binary Step function}. Note that the approach presented is specifically tailored for NNs with PWA \textit{activation functions}. The exploration to accommodate nonlinear \textit{activation functions} will be a subject of future research. The \textit{activation function} $\sigma^{(l)} : \mathbb{R}^n \rightarrow \mathbb{R}^n$ of a layer $l$ with $n$ neurons can then be written as:
\begin{align*}
\sigma^{(l)}(x) = [\sigma(x_1), \hdots, \sigma(x_n)]^T = C_{\Phi(x)}^{(l)} x + d_{\Phi(x)}^{(l)}
 ,\forall x\in \mathbb{R}^n
\end{align*}
where $x = (x_1,x_2, \hdots ,x_n)^T\in \mathbb{R}^n$, $C_{\Phi(x)}^{(l)} \in \mathbb{R}^{n \times n}$ is a square diagonal matrix and $d_{\Phi(x)}^{(l)} \in \mathbb{R}^{n}$. 

Note that the elements of $C_{\Phi(x)}^{(l)}$ and $d_{\Phi(x)}^{(l)}$ are the first order coefficients of $\sigma$ for $x_1, x_2, ... , x_n$ respectively, i.e. $\forall i \in \{1, \hdots, n\}$,
\begin{align*}
    \exists j \text{ s.t. } m_{j-1} < x_i \leq m_j \text{ and}
    \begin{cases}
        C_{{\Phi(x)}^{i,i}}^{(l)} = c_j\\
        d_{{\Phi(x)}^{i}}^{(l)} = d_j
    \end{cases}
\end{align*}
The subscript $\Phi(x)$ will throughout this paper represent the index of the affine region of the NN when given an input $x$. Assume that $x \in [\underline{m},\overline{m}]$, i.e, $\underline{m}_i<x_i \leq \overline{m}_i$, $i=1,\ldots,n$, where $\underline{m}, \overline{m} \in \mathbb{R}^n$. We refer to $\underline{m}, \overline{m}$ as \textit{delimiters} of the vector $x$. These vectors essentially define the boundary points for each element in $x$ and when an individual element crosses these boundaries, $C_{\Phi(x)}^{(l)}$ and $d_{\Phi(x)}^{(l)}$ are altered accordingly.\\
Given an input $x$, the output of layer $l$ can be written as:
\begin{align}
    z^{(l)} &= \sigma^{(l)}(W^{(l)}z^{(l-1)} + b^{(l)}) \nonumber \\ &=(C_{\Phi(x)}^{(l)}W^{(l)})z^{(l-1)} + (C_{\Phi(x)}^{(l)}b^{(l)} + d_{\Phi(x)}^{(l)})
    \label{eq: cw}
\end{align}
Thus, we can write the output of layer $l$ w.r.t. the input $x$:
\begin{align}
    z^{(l)} &=  \mathcal{E}^{(l)}_{\Phi(x)}x+\mathcal{G}^{(l)}_{\Phi(x)}  \text{ , where} \label{activ_para}\\
    &\mathcal{E}^{(l)}_{\Phi(x)}={\overset{\curvearrowleft}\prod^{l}_{k = 1}} (C_{\Phi(x)}^{(k)}W^{(k)})\nonumber\\
    &\mathcal{G}^{(l)}_{\Phi(x)} = \sum_{k = 1}^{l} [{\overset{\curvearrowleft}\prod^{l}_{j=k+1}} C_{\Phi(x)}^{(j)}W^{(j)}](C_{\Phi(x)}^{(k)}b^{(k)} + d_{\Phi(x)}^{(k)})\nonumber
\end{align} 
We call $\mathcal{E}^{(l)}_{\Phi(x)}$ and $\mathcal{G}^{(l)}_{\Phi(x)}$ the \textit{active parameters} of $x$.
\begin{defn}
    Let $\mathcal{N}$ be a feed-forward NN, utilizing solely PWA \textit{activation functions}, defined by (\ref{activation}). Let $\mathcal{X}$ be the set of all possible inputs to $\mathcal{N}$. A \textit{linear region} $\mathcal{R}$ of the layer $(l)$ of $\mathcal{N}(x;\theta)$ is defined as the set:
    \begin{align*}
        \mathcal{R} = 
        \big\{x\in \mathcal{X} \mid
        \underbar{$m$} \preccurlyeq  W^{(l)}z^{(l-1)} + b^{(l)} \preccurlyeq \overline{m}, \\z^{(l-1)}=\mathcal{E}^{(l-1)}_{\Phi(x)} x+\mathcal{G}^{(l-1)}_{\Phi(x)} \big\}
    \end{align*}
    where $\underbar{$m$}$, $\overline{m}$ are the \textit{delimiters} of $W^{(l)}z^{(l-1)}+b^{(l)}$ in $\sigma^{(l)}$.
\end{defn}
It is proven that each \textit{linear region} is a convex polytope \cite{chu2018exact}. Intuitively, a \textit{linear region} is the convex set of states that share the same \textit{active parameters}. Moreover, if the PWA \textit{activation functions} are all continuous, the output will be continuous as well. The closed-loop system in (\ref{eq: sys_NN}) becomes:
\begin{align}
    \dot{x} = Ax + B \mathcal{N}(x; \theta) 
    = (A + B \mathcal{E}^{(L)}_{\Phi(x)}) x + B \mathcal{G}^{(L)}_{\Phi(x)}
    \label{eq: closed_loop}
\end{align}
Note that $\Phi(x) = j$ if $x \in \mathcal{R}_j$ where $R_j$ is a polytope corresponding to a \textit{linear region} and is defined as $\mathcal{R}_j=\{z \mid C_jz\leq d_j\}$. Hence the dynamic of the closed-loop system is continuous PWA and will change when the state crosses into another \textit{linear region}. 

\subsection{Invariance condition}
For the PWA dynamical system to remain inside a safe set according to Nagumo theorem also \cite[Theorem 4.7]{inbookcite}, the vector field has to point tangentially or inward at every point along the set's boundary.  Considering our problem statement, we can take advantage of the linear constraints of the polytopic \textit{linear regions} and the linearity of the closed-loop system within these regions to simplify this condition. Indeed, we show that it is sufficient to examine solely the vertices of the set we want to verify and the points on its border where a transition between \textit{linear regions} occurs. 
\begin{thm}
Consider a polytopic set $\mathcal{S} = \{x \in \mathbb{R}^n\mid C^{\mathcal{S}}x \leq d^{\mathcal{S}} \}$ and let $\mathcal{F}(\mathcal{S}) = \{\mathcal{F}_1, \hdots, \mathcal{F}_N\}$ the faces of $\mathcal{S}$. Consider the linear switched system $\Sigma$ in (\ref{eq: closed_loop}) defined by $\dot x = f(x) = A_{\Phi(x)}x + b_{\Phi(x)}$ where $\Phi(x) = j$ if $x \in \mathcal{R}_j = \{z \mid C_jz\leq d_j\}$, $1\leq j\leq M$. Furthermore, assume $\mathcal{S} \subseteq \bigcup_{j=1}^{M} \mathcal{R}_{j}$.  Consider the set
\begin{align}
    \mathcal{D}_{ij} = \mathcal{F}_i \cap \mathcal{R}_j , i \in \{1, \hdots, N\}, j \in \{1, \hdots, M\}
    \label{eq: B}
\end{align}
Then, for any trajectory $ x : \mathbb{R}_0^{+} \rightarrow \mathbb{R}^n$ satisfying $x(0) \in \mathcal{S}$, we have
$x(t) \in \mathcal{S}, \forall t \in \mathbb{R}_0^+$ for the system $\Sigma$ if and only if the following condition is satisfied:
\begin{align}
    &C^{\mathcal{S}}_{i,*} \cdot (A_j v + b_j) \leq 0, \nonumber \\ &\forall i \in \{1,\hdots,N\}, \forall j \in \{1,\hdots, M\}, \forall v \in \mathcal{V}(\mathcal{D}_{ij})
    \label{invariance_condition}
\end{align}
\label{thm: vertex}
\end{thm}
\begin{proof}
    Let us rewrite $\mathcal{S} = \{x \in \mathbb{R}^n\mid g^{\mathcal{S}}_i(x)\leq 0 , i \in \{1, \hdots, N\}\}$ where $g^{\mathcal{S}}_i:\mathbb{R}^n \rightarrow \mathbb{R}$ defined as $g^{\mathcal{S}}_i(x)=C^{\mathcal{S}}_{i,*}x - d_i$. Consider a point $x \in \partial \mathcal{S}$. Then there exists $i\in\{1,\hdots,N\}$, such that $x\in \mathcal{F}_i(\mathcal{S})$. Using the fact that $\mathcal{S} \subseteq \bigcup_{j=1}^{M} \mathcal{R}_{j}$ there exists $j \in \{1, \hdots, M\}$ such that $x \in \mathcal{D}_{ij}$. Hence, one gets 
    \begin{align*}
            x = \sum_{k = 1}^{l} \lambda_k v_k \text{ where } \lambda_k \in [0, 1], \sum_{k = 0}^{l} \lambda_k = 1, v_k \in \mathcal{V}(\mathcal{D}_{ij})
    \end{align*}
    where $l$ is the number of vertices of the polytope $\mathcal{D}_{ij}$. Thus,
    \begin{align*}
        &\nabla g^{\mathcal{S}}(x)^T f(x)\\
        &=(C^{\mathcal{S}}_{i,*})^T\cdot(A_j x + b_j)
        =(C^{\mathcal{S}}_{i,*})^T\cdot(A_j\sum_{k = 0}^{l} \lambda_k v_k + b_j)\\
        &= \sum_{k = 0}^{l} \underbrace{\lambda_{k}[(C^{\mathcal{S}}_{i,*})^T \cdot (A_jv_k + b_j)}_{\leq 0}]
        \leq 0
    \end{align*}
implying that $f(x) \in T_S(x)$\footnote{$T_{\mathcal{S}}(x)$ denotes the tangent cone to the set $\mathcal{S}$ at the point $x$ (see \cite[Definition 4.6]{inbookcite}).}. It follows from \cite[Theorem 4.7]{inbookcite} that $\mathcal{S}$ is invariant for the considered system. Now assume that $\mathcal{S}$ is  an invariant for the system $\Sigma$. We have from \cite[Theorem 4.7]{inbookcite} that $f(x) \in T_S(x)$ for all $x \in \partial S$. Now consider $v_{ij} \in \mathcal{D}_{ij}$, $i \in \{1,\ldots,N\}$, $j \in \{1,\ldots,M\}$. Since $v_{ij} \in \partial S$, it follows that $f(v_{ij}) \in T_S(v_{ij})$, which in turn implies (\ref{invariance_condition}). 
\end{proof}
\begin{figure}
    \centering
    \includegraphics[width = 0.28\textwidth]{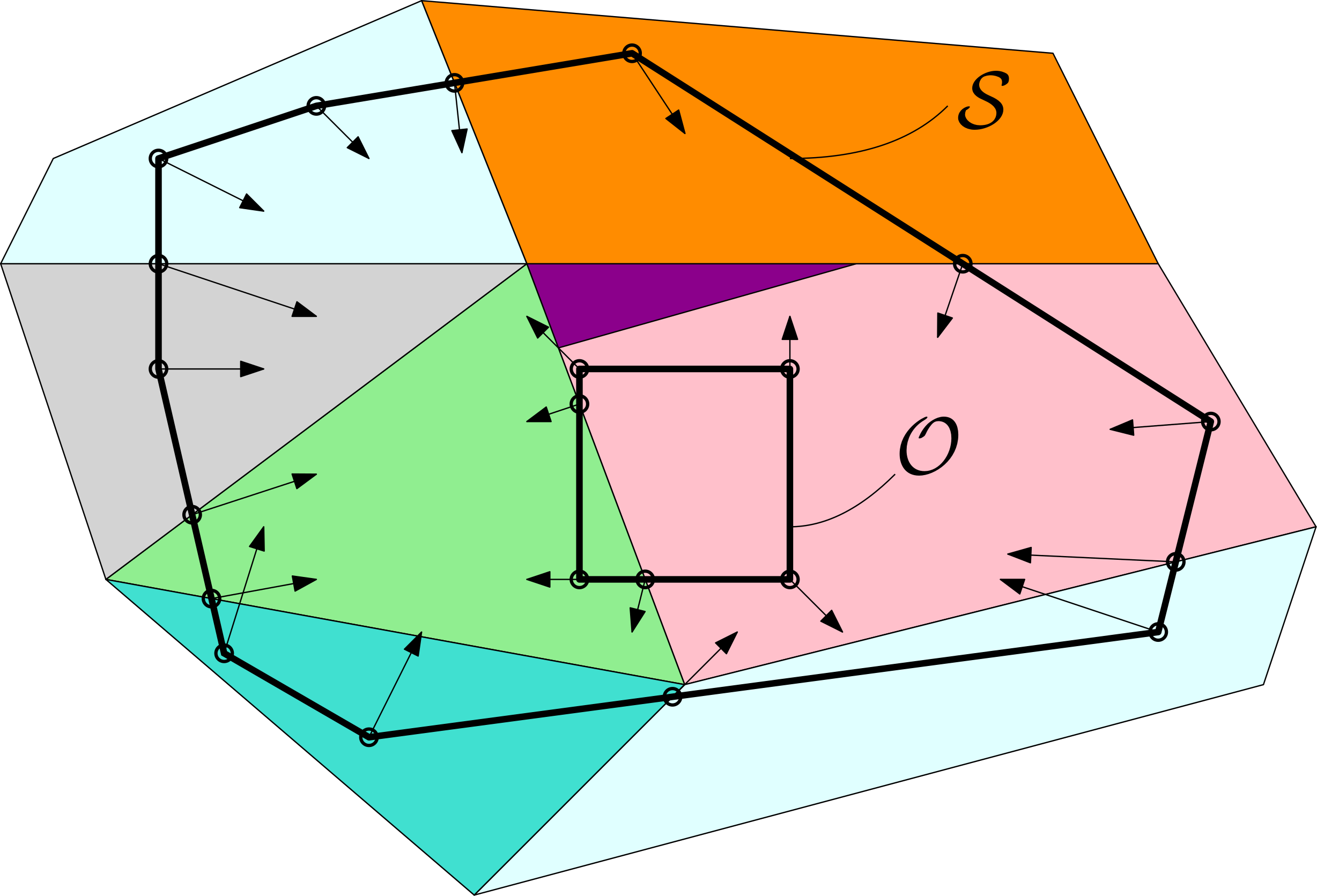}
    \caption{Every \textit{linear region} has a different color. A finite set of vertices must be verified to assess the invariance of $\mathcal{S}\setminus\mathcal{O}$.}
    \label{fig: corollary}
\end{figure}
For an open convex polytopic unsafe set $\mathcal{O} = \{x \in \mathbb{R}^n\mid C^{\mathcal{O}}x < d^{\mathcal{O}} \}$, the inequality of equation (\ref{invariance_condition}) changes, such that for any trajectory $ x : \mathbb{R}_0^{+} \rightarrow \mathbb{R}^n$ satisfying $x(0) \notin \mathcal{O}$, we have
$x(t) \notin \mathcal{O}, \forall t \in \mathbb{R}_0^+$ for the system $\Sigma$ if and only if the following condition is satisfied:
\begin{align}
    &C^{\mathcal{O}}_{i,*} \cdot (A_j v + b_j) \geq 0, \nonumber \\ &\forall i \in \{1,\hdots,N\}, \forall j \in \{1,\hdots, M\}, \forall v \in \mathcal{V}(\mathcal{D}_{ij})
    \label{invariance_condition_cor}
\end{align}
We call (\ref{invariance_condition}) and (\ref{invariance_condition_cor}) the \textit{invariance condition}. An intuitive 2-dimensional example is depicted in Figure \ref{fig: corollary}.

\section{Algorithmic Implementation}
\label{sec:4}
\subsection{Segmentation into Linear Regions}
The division of the state space into \textit{linear regions} comes from the \textit{breakpoints} in the PWA \textit{activation function} of each neuron. The hyperplane separating two regions can be written w.r.t the output of the previous layer $z^{(l-1)}$:
\begin{equation}
    \{z^{(l-1)}\in \mathbb{R}^{n^{(l-1)}} \mid W^{(l)}_{n,*} z^{(l-1)} + b^{(l)}_n = m_k\}
    \label{eq: hyperplane_z}
\end{equation}
where $W^{(l)}_{n,*}$ are the weights linked to neuron $n$ in layer $l$, $b^{(l)}_n$ its bias and $m_k$ the $k^{th}$ breakpoint of its \textit{activation function}.
In a multi-layer NN, every layer of the network will cut every region $\mathcal{R}_j$ coming from the previous layer with (\ref{eq: hyperplane_z}). As the \textit{linear regions} should be expressed w.r.t. the state space, we can use (\ref{activ_para}) to rewrite (\ref{eq: hyperplane_z}) as:
{\begin{align}
    &\mathcal{H}^{(l)}_{n,k}\{\mathcal{R}_j\}= \nonumber\\
    &\{x \in \mathbb{R}^{n_{in}} \mid W^{(l)}_{n,*}[\mathcal{E}_{j}^{(l-1)}x + \mathcal{G}_{j}^{(l-1)}] + b_n^{(l)} = m_k\}
    \label{eq: hyperplane}
\end{align}}
This hyperplane represents the border between two \textit{linear regions} of layer $(l)$ within region $\mathcal{R}_j$. Note that the hyperplane depends on the \textit{active parameters} of the \textit{linear region} it cuts. As a result, the hyperplane is only valid inside its associated \textit{linear region}. Equation (\ref{eq: hyperplane}) also shows that the computation of a region of a given layer is dependent on a region of the previous layer. Consequently, we compute all the \textit{linear regions} of a layer before moving to the next layer. 

The \textit{active parameters} can be computed using an iterative formula of (\ref{activation}). However, modern deep-learning libraries (e.g. \textit{PyTorch} \cite{paszke2019pytorch}) build a computation-graph using \textit{automatic differentiation} \cite{baydin2018automatic} making the computation of the gradient of any element in the network w.r.t. another possible and efficient. As the output of a layer inside a \textit{linear region} has by definition a constant derivative, the Jacobian matrix only needs to be evaluated for a single point inside the region. We take the Chebyshev center $x_c$, i.e. the center of the largest Euclidean ball that is entirely contained within the polytope, because it is by definition inside the region and its computation scales well for higher dimensions. We can express the \textit{active parameters} in a more simple way than (\ref{activ_para}):
\begin{align*}
    \mathcal{E}^{(l)}_{\Phi(x_c)} &= \left. \frac{\partial z^{(l)}}{\partial x}\right|_{x_c};
    \mathcal{G}^{(l)}_{\Phi(x_c)} &= \left. z^{(l)}\right|_{x_c} - \left. \frac{\partial z^{(l)}}{\partial x}\right|_{x_c} \cdot x_c
\end{align*}
Figure \ref{fig: div} depicts the segmentation algorithm. It takes as input the NN $\mathcal{N}$, a closed polytopic set $\mathcal{S} \subset \mathbb{R}^{n_{in}}$ and $\mathcal{O} = \{\mathcal{O}_1, \hdots, \mathcal{O}_m\} \subset \mathcal{S}$ a list of open polytopic sets and returns the list of \textit{linear regions} of the state-space.
Note that the segmentation algorithm will provide a unique decomposition of the set $\mathcal{S}\setminus\mathcal{O}$ under a given NN. Figure \ref{fig:tree-graph} shows how the segmentation can be seen as a tree-graph. If we consider $\mathcal{R}_4$ of no interest, we can ignore the computation of all its subregions and hence save considerable time. Thus, in our approach, a substantial portion of the computation time can be saved since we do not need to process \textit{linear regions} that do not intersect the border of the set to be analysed. An example would be the purple region in Figure \ref{fig: corollary}.
\begin{figure}
    \centering
    \includegraphics[width = 0.48\textwidth]{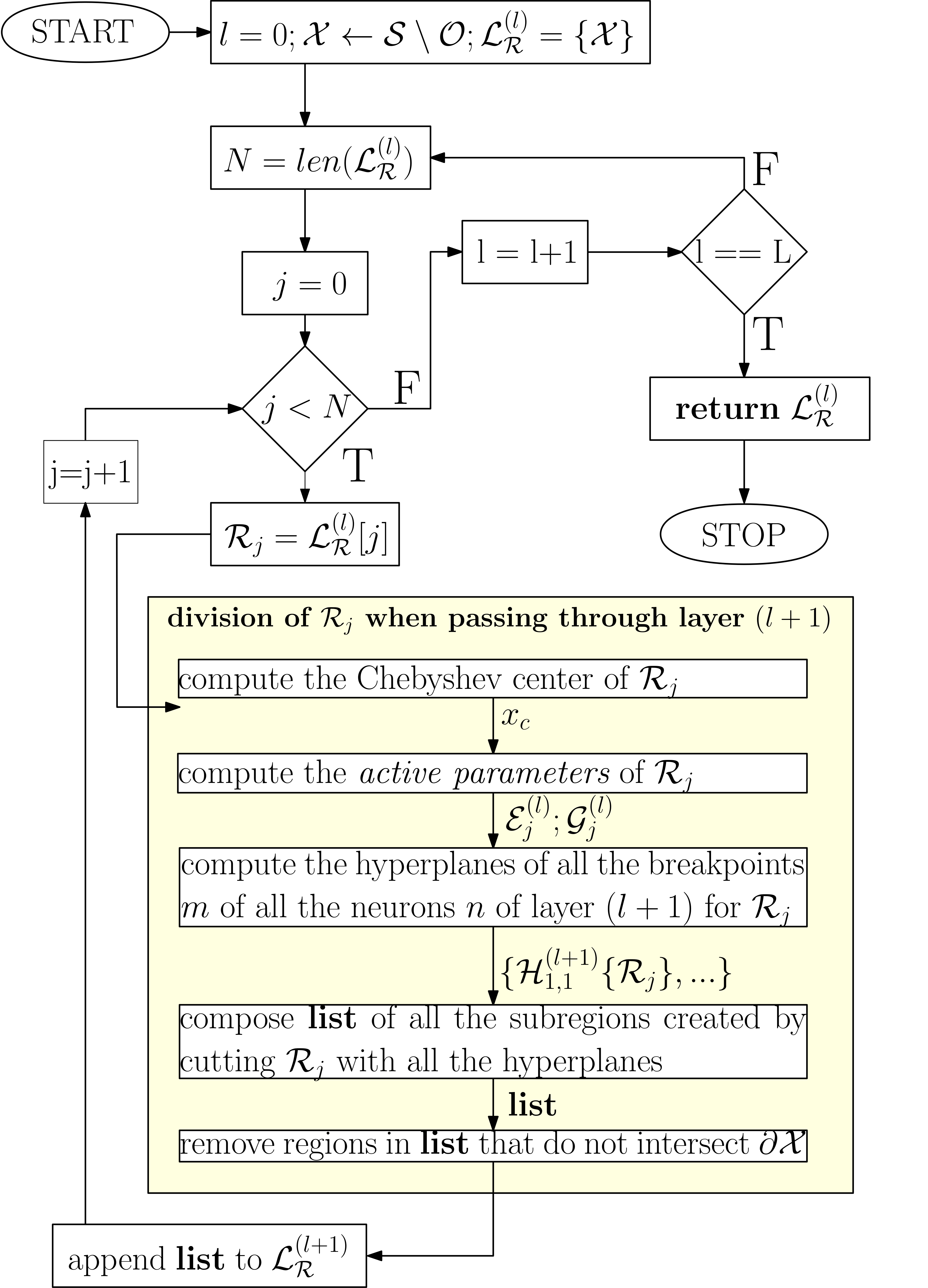}
    \caption{Segmentation algorithm: $\mathcal{L}_\mathcal{R}^{(l)}$ is the list of all the regions of layer $l$. The algorithm cuts every region of $\mathcal{L}_\mathcal{R}^{(l)}$ with their associated hyperplanes of layer $l+1$.}
    \label{fig: div}
\end{figure}
\begin{figure}
    \centering
    \includegraphics[width = 0.48\textwidth]{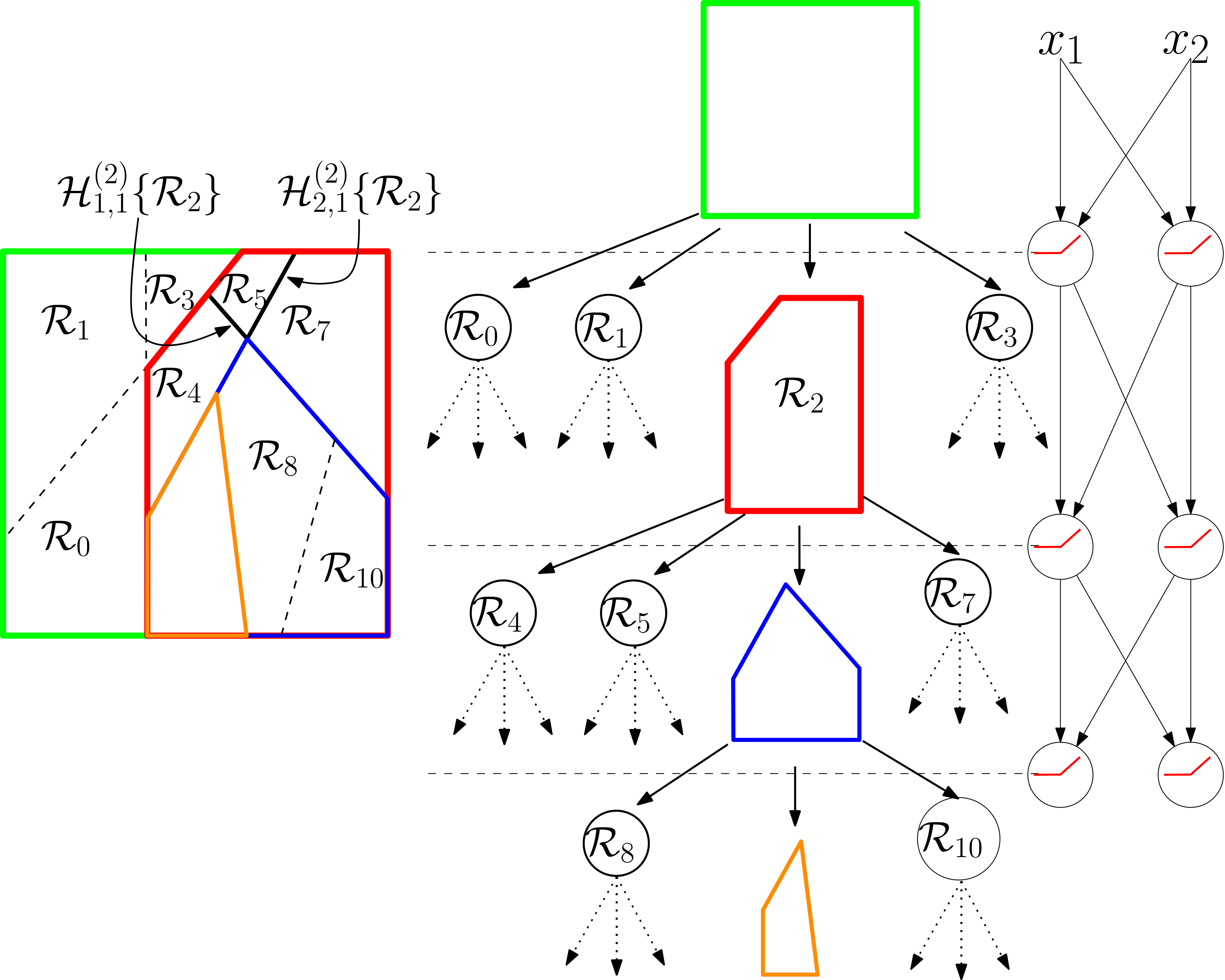}
    \caption{Segmentation of the green polytope in $\mathbb{R}^2$ by a NN having for every layer $l$, $\sigma^{(l)}(x) = max(0,x)$. When a region $\mathcal{R}_j$ of layer $(l)$ is being cut by the forthcoming layer $(l+1)$, we first compute the \textit{active parameters} $\mathcal{E}_{j}^{(l)}$ and $\mathcal{G}_{j}^{(l)}$  of the region. These are then used to compute all the hyperplanes $\mathcal{H}^{(l+1)}\{\mathcal{R}_j\}$ in layer $(l+1)$. Finally, these hyperplanes cut (only) the region $\mathcal{R}_j$ in subregions.}
    \label{fig:tree-graph}
\end{figure}

\subsection{Safety Verification}
Once we have divided the state-space into \textit{linear regions}, we use Algorithm \ref{algo: safety} to verify the safety of the closed-loop system. It assesses whether the vector field guides the system back inside the safe set for each vertex of the \textit{invariance condition}.
 \begin{algorithm}[h]
 \small
 \caption{Verifying the safety of a NN-controller}
 \begin{algorithmic}[1]
 \renewcommand{\algorithmicrequire}{\textbf{Input:}}
 \renewcommand{\algorithmicensure}{\textbf{Output:}}
 \REQUIRE neural network $\mathcal{N}$, closed polytopic set $\mathcal{S} \subset \mathbb{R}^{n_{in}}$, $\mathcal{O} = \{\mathcal{O}_1, \hdots, \mathcal{O}_n\} \subset \mathcal{S}$ a list of open polytopic sets and the closed-loop dynamical system $\Sigma$ in (\ref{eq: sys_NN}).
\ENSURE a boolean \texttt{safety}
\\ \textit{Initialisation} : $\texttt{safety} \leftarrow \texttt{TRUE}$ 
\STATE $\mathcal{L}_\mathcal{R} \leftarrow \texttt{segmentation\_algorithm}(\mathcal{N}, \mathcal{S}, \mathcal{O})$
\FOR {\textbf{each polytope }$\mathcal{R}$ \textbf{in} $\mathcal{L}_\mathcal{R}$}
    \FOR {\textbf{each face }$\mathcal{F}_i$ \textbf{in} $\mathcal{F}(\mathcal{S})$}
    \STATE $\mathcal{P} = \mathcal{R} \cap \mathcal{F}_i$
    \IF{$\mathcal{P} \neq \emptyset$}
        \FOR {\textbf{each vertex} $v$ in $\mathcal{V}(\mathcal{P})$}
                    \IF{$C^{\mathcal{S}}_{i,*}f(v) > 0$}
                         \STATE $\texttt{safety} \leftarrow \texttt{FALSE}$
                    \ENDIF
        \ENDFOR
    \ENDIF
    \ENDFOR
\FOR {\textbf{each polytope }$\mathcal{O}_k$ \textbf{in} $\mathcal{O}$}
    \FOR {\textbf{each face }$\mathcal{F}_i$ \textbf{in} $\mathcal{F}(\partial\mathcal{O}_k)$}
    \STATE $\mathcal{P} = \mathcal{R} \cap \mathcal{F}_i$
    \IF{$\mathcal{P} \neq \emptyset$}
        \FOR {\textbf{each vertex} $v$ in $\mathcal{V}(\mathcal{P})$}
                    \IF{$C^{\mathcal{O}}_{i,*}f(v) < 0$}
                         \STATE $\texttt{safety} \leftarrow \texttt{FALSE}$
                    \ENDIF
        \ENDFOR
    \ENDIF
    \ENDFOR
\ENDFOR
\ENDFOR
\RETURN \texttt{safety}
\end{algorithmic}
\label{algo: safety}
\end{algorithm}

\section{Numerical Examples}
\label{sec:5}
\subsection{Mobile robot}
We consider a 2-dimensional environment in which an agent navigates. The agent behaves as an integrator and is represented by: $\dot{x}_{2\times1} = u_{2\times1}$.  The safety of the controlled system will be verified for $\mathcal{S}_1 = \{(x,y)\mid -5\leq x\leq 5, -5\leq y \leq 5\}$ and two polytopic unsafe sets $\mathcal{O}_1$ and $\mathcal{O}_2$. Figure {\ref{fig: visu}} depicts the environment and the \textit{linear regions} identified by the segmentation algorithm illustrated in Figure \ref{fig: div}. The NN is trained using Reinforcement Learning, more specifically DDPG \cite{lillicrap2015continuous}, to reach a target position while avoiding obstacles.
The architecture of the NN controller is $2\times128^{(2)}\times 2$. The activation function utilized in the hidden layers is the $\textit{leaky-relu}(x) = \max(0.01x, x)$ and on the last layer we use a linearized version of $tanh(x)$.
\begin{figure}[h]
     \centering
     \begin{subfigure}[b]{0.23\textwidth}
         \centering
         \includegraphics[width=\textwidth]{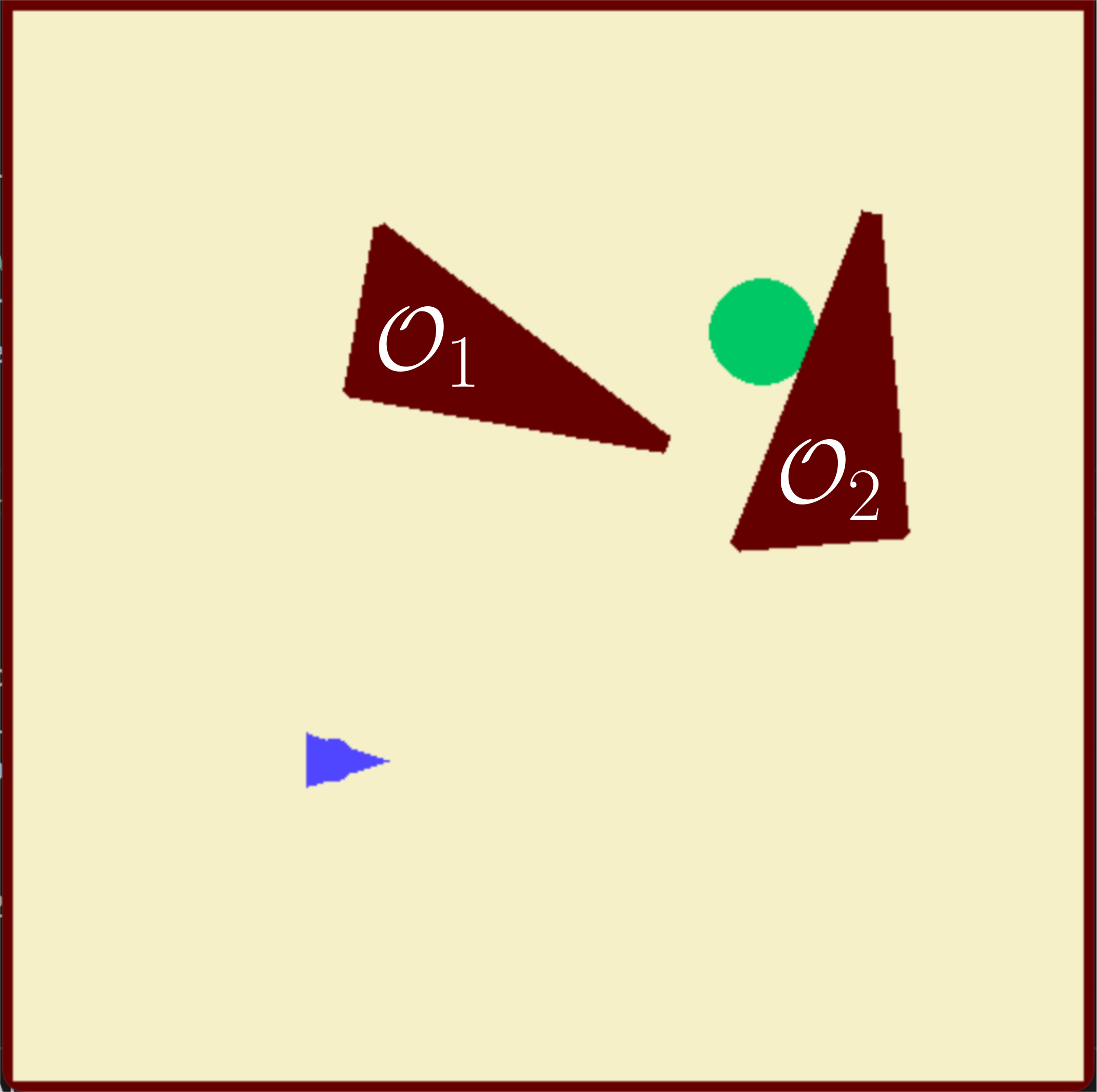}
         \caption{2D environment with the agent in blue and the unsafe polytopic sets $\mathcal{O}_1$ and $\mathcal{O}_2$ in red}
         \label{fig:environment}
     \end{subfigure}
     \hspace{0.2cm}
     \begin{subfigure}[b]{0.23\textwidth}
         \centering
         \includegraphics[width=\textwidth]{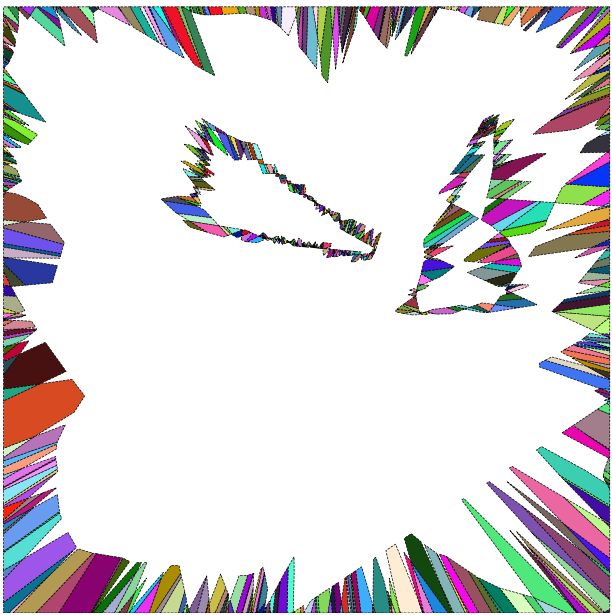}
         \caption{Segmentation of the state space into 1002 regions by a $2\times 128^{(2)}\times2$ NN controller}
         \label{fig: subregions}
     \end{subfigure}
     \caption{Visualization of the \textit{Mobile Robot}}
     \label{fig: visu}
\end{figure}
\subsection{Spring-mass-damper}
The system comprises $n$ interconnected wagons, with each wagon having two variables: position and velocity. Figure \ref{fig:mass_spring_damper} depicts the system. The system's representation is as follows:
\begin{align*}
    \dot{x} = 
    \begin{bmatrix}
        [0]_{n\times n} & I_{n}\\
        -M^{-1}K & M^{-1}C\\
    \end{bmatrix}
    x + 
    \begin{bmatrix}
        [0]_{n\times n}\\
        M^{-1} \\
    \end{bmatrix} [F_{u_1}, \hdots, F_{u_n}]^T
\end{align*}
where
\begin{align*}
    K = \left [
    \begin{smallmatrix}
        (k_1 + k_2) & -k_2 & 0 & 0& \hdots & 0\\
        -k_2 & (k_2+k_3) & -k_3 & 0 & \hdots & 0\\
        0 & -k_3 &  (k_3 + k_4) & -k_4 & \hdots & 0\\
        \vdots &\ddots &\ddots &\ddots &\ddots & \vdots\\
        0 & \hdots & 0& -k_{(n-1)} & (k_{(n-1)} + k_{n}) & -k_n\\
        0 & \hdots & \hdots  & 0& - k_n & k_n\\
    \end{smallmatrix}\right ],
\end{align*}
$M = diag[m_1, m_2, \hdots, m_n]$, $x = [z_1, \hdots, z_n, \dot z_1, \hdots, \dot z_n]$, where $z_i$ is the position of wagon $i$, $i=1,\ldots,n$ and $C$ has the same structure as $K$. The system is well-suited for the investigation of the scalability of our approach w.r.t. the dimensionality of the systems. Algorithm \ref{algo: safety} is applied to the set $\mathcal{S}_2 = \{ x\in \mathbb{R}^{2n}\mid 0 \leq x_i \leq 1, |x_{n+i}|\leq |x_i|,  i = 1, 2, \hdots, n\}$.
We train the NN by approximating an MPC controller. During training, the initial state is chosen from $\{ x\in \mathbb{R}^{2n}\mid |x_i| <1,  i = 1, \hdots,n\}$ and we impose $-1\leq u_i\leq 1, i= 1,\hdots, n$.
\begin{figure}
    \centering
    \includegraphics[width = 0.48\textwidth]{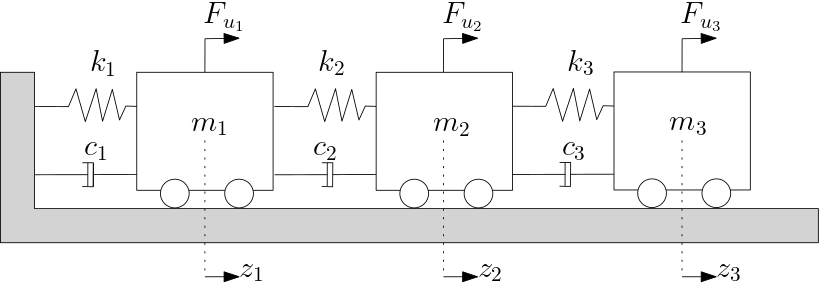}
    \caption{\textit{Spring-mass-damper} with 3 wagons, i.e. 6 states}
    \label{fig:mass_spring_damper}
\end{figure}
\subsection{Scalability}

The investigation of the time complexity relies on estimating the count of \textit{linear regions} in a multi-layer, multi-breakpoint NN, which unfortunately remains an open research question \cite{montufar2014number, hanin2019complexity}. Nevertheless, we aim to provide a clear idea of how computational time scales.

Considering the segmentation algorithm, the worst-case time complexity for partitioning a region by a layer is directly related to the maximum number of subregions that a set of $k$ hyperplanes can partition a $d$-dimensional space. We can define $K(n^{(l)};d)$ as the time complexity associated with segmenting a region $\mathcal{R}$ with dimensionality $d$ by layer $l$:
\begin{align*}
    K(n^{(l)};d) = 
    O\left(1 + n^{(l)} + \binom{n^{(l)}}{1} + \hdots + \binom{n^{(l)}}{d}\right)
\end{align*}
Let $\#\mathcal{R}(l)$ represent the count of regions in layer $l$ that intersect the border of the verification set. Consequently, the time complexity of the segmentation algorithm becomes:

\footnotesize
\begin{align*}
    \centering
    O\left(K(n^{(1)};d) + K(n^{(2)};d) \#\mathcal{R}(1)  + \hdots + K(n^{(L)};d) \#\mathcal{R}(L-1)\right)
\end{align*}
\normalsize

The time complexity of Algorithm 1 is contingent upon the count of \textit{linear regions} identified and their number of vertices. The time complexity of the vertex enumeration problem amounts to $O(f^2dv)$, where $f$ denotes the number of faces of the region and $v$ signifies the number of vertices. Importantly, the number of vertices exhibits exponential growth with the dimensionality $d$. The time complexity of verifying all the vertices becomes: $O\left(\#\mathcal{R}(L) \cdot 2^d\right)$.

Due to the absence of a formal estimation concerning the count of \textit{linear regions}, we conduct empirical tests. We investigate the effects of varying the depth and width of the NN and the dimension of the system. The data presented in Table \ref{table: width}, \ref{table: depth}, and \ref{table: mpc} includes essential metrics: $\#N$, $\#\theta$, $\#\mathcal{R}$, and $t_v$ representing the number of hidden neurons, the number of network parameters, the number of \textit{linear regions} and the verification time of Algorithm \ref{algo: safety}, respectively. Table \ref{table: width} and \ref{table: depth} provide results from experiments conducted on the \textit{Mobile Robot}, while Table \ref{table: mpc} from the \textit{Spring-mass-damper} system.
\begin{threeparttable}[htb]
    \footnotesize
    \setlength\tabcolsep{0pt}
\begin{tabular*}{\linewidth}{@{\extracolsep{\fill}} l cc cc cc @{}}
&   \multicolumn{4}{c}{} \\
Architecture           & $\#N$
                            & $\#\theta$
                                & $\#\mathcal{R}$
                                    & $t_v$ [s] \\ 
        \midrule
$2\times 16^{(2)} \times 2$  &32   &354&       125  & 28  \\
$2\times 32^{(2)} \times 2$  &64   &1218&      248  & 92  \\
$2\times 64^{(2)} \times 2$  &128  &4482&      477 & 307  \\
$2\times 128^{(2)}\times 2$  &256  &17154&     973 &  1263 \\
$2\times 256^{(2)} \times 2$ &512  &67074&     1835 & 4541  \\
\end{tabular*}
\caption{Scalability w.r.t. the NN's width}
\label{table: width}
\end{threeparttable}
\begin{threeparttable}[htb]
    \footnotesize
    \setlength\tabcolsep{0pt}
\begin{tabular*}{\linewidth}{@{\extracolsep{\fill}} l cc cc cc @{}}
&   \multicolumn{4}{c}{} \\
 Architecture           & $\#N$
                            & $\#\theta$
                                & $\#\mathcal{R}$
                                    & $t_v$ [s] \\ 
        \midrule
$2\times 32^{(1)} \times 2$  & 32  &162& 122    &   31\\
$2\times 32^{(2)} \times 2$  & 64  &1218& 248    & 92  \\
$2\times 32^{(4)} \times 2$  & 128 &3330& 412 &  263 \\
$2\times 32^{(6)} \times 2$  & 192 &5442& 597 &  566 \\
$2\times 32^{(8)} \times 2$  & 256 &7554& 862   &1185   \\
\end{tabular*}
\caption{Scalability w.r.t. the NN's depth}
\label{table: depth}
\end{threeparttable}
\begin{threeparttable}[htb]
    \footnotesize
    \setlength\tabcolsep{0pt}
\begin{tabular*}{\linewidth}{@{\extracolsep{\fill}} l cc cc cc @{}}
&   \multicolumn{4}{c}{} \\
Architecture          & $\#N$
                            & $\#\theta$
                                & $\#\mathcal{R}$
                                    & $t_v$ [s] \\ 
        \midrule
$2\times 8^{(2)} \times 1$     & 16   &105     &37      & 6  \\
$4\times 8^{(2)} \times 2$     & 16   &130     &253     & 62  \\
$6\times 8^{(2)} \times 3$     & 16   &155     &1937    & 384  \\
$8\times 8^{(2)} \times 4$     & 16   &180     &8120       & 1629  \\
\end{tabular*}
\caption{Scalability w.r.t. the system's dimension}
\label{table: mpc}
\end{threeparttable}
When comparing, it is important to clarify that our approach is currently tailored to linear systems, while existing barrier-based approaches \cite{zhao2020synthesizing, zhao2022verifying} can extend to nonlinear systems. However, their numerical examples are limited to systems of dimension 3.
Our approach exhibits favorable scalability when it comes to network size and we demonstrate that although our method is limited to linear systems, we manage to treat systems with dimensionality up to 8 in a reasonable time on an ordinary machine\footnote{All experiments were conducted on a 2016 MacBook Pro with 16 GB RAM, a 3.3 GHz Intel Core i7 dual-core processor. All the processing was done on the CPU. The codebase (Python) is available at \url{https://github.com/LouisJouret/Neural-Control-Invariance-Checker}.}. Our approach provides finite-time safety assessment, a capability lacking in existing barrier-based methods \cite{zhao2020synthesizing, zhao2022verifying}. Furthermore, if the computational time of these methods becomes excessive, they may require a trial-and-error procedure consisting of successive attempts with different NN architectures until potentially a successful NN-barrier-function is found. Conversely, our approach scales in a more predictable manner with the number of neurons and the dimensionality of the system.

\section{Conclusion and Future Work}

This paper presents a new method for evaluating the safety of linear dynamical systems controlled by a NN within a potentially non-convex region of the state space. Our algorithm is complete and exhibits excellent scalability properties regarding the number of neurons and the system's dimensionality. Our future research will concentrate on extending this approach to nonlinear systems. This will entail approximating the nonlinear system using a NN with PWA activation functions, effectively converting the closed-loop system into a linear representation.

\printbibliography{}

\end{document}